\documentclass[11pt,onecolumn,draftcls,peerreview]{IEEEtran}
\usepackage{amssymb}
\usepackage{amsmath}
\usepackage{amsfonts}
\usepackage{float}
\usepackage{graphics}
\usepackage[tight,footnotesize]{subfigure}
\usepackage{rotating}
\usepackage{algorithm}
\usepackage{algorithmic}
\usepackage{enumerate}
\usepackage{cite}
\usepackage{stfloats}
\usepackage{multicol}
\usepackage{array}

\newcommand*{\QEDA}
{\hfill\ensuremath{\blacksquare}}
\DeclareMathSizes{10}{9}{7}{5}
\newtheorem{theo}{Theorem}

\newtheorem{theorem}{Theorem}

\newtheorem{corollary}{Corollary}

\newtheorem{example}[theorem]{Example}

\newtheorem{remark}{Remark}

\input{tcilatex}
\begin{document}

\title{Artificial Noise Revisited}
\author{Shuiyin~Liu,~Yi~Hong,~and~Emanuele~Viterbo\thanks{%
S.~Liu, Y.~Hong and E.~Viterbo are with the Department of Electrical and
Computer Systems Engineering, Monash University, Clayton, VIC 3800,
Australia (e-mail: \{shuiyin.liu, yi.hong, emanuele.viterbo\}@monash.edu).
This work was performed at the Monash Software Defined Telecommunications
Lab and the authors were supported by the Monash Professorial Fellowship,
2013 Monash Faculty of Engineering Seed Funding Scheme, and the Australian
Research Council Discovery Project with ARC DP130100336.}}
\maketitle

\begin{abstract}
The artificial noise (AN) scheme, proposed by Goel and Negi, is being
considered as one of the key enabling technology for secure communications
over MIMO wiretap channels. However, the decrease in secrecy rate due to the
increase in the number of Eve's antennas is not well understood. In this
paper, we develop an analytical framework to characterize the secrecy rate
of the AN scheme as a function of Eve's SNR, Bob's SNR, the number of
antennas in each terminal, and the power allocation scheme. We first derive
a closed-form expression for the average secrecy rate. We then derive a
closed-form expression for the asymptotic instantaneous secrecy rate with
large number of antennas at all terminals. Finally, we derive simple lower
and upper bounds on the average/instantaneous secrecy rate that provide a
tool for the system design.
\end{abstract}


\section{Introduction}

The security of data transmissions is a fundamental issue in wireless
communication systems, where the broadcast characteristics make it difficult
to prevent eavesdropping. Traditional key-based cryptography \cite{Hellman76,Goldwasser82,Silverman98} usually is based on the assumption that
the eavesdropper (Eve) has limited computational resources. These algorithms
ensure that it is computational infeasible to decipher the encrypted
messages without knowledge of the secret key. On the other hand, assuming
Eve has unlimited computational power, Wyner, in \cite{Wyner75}, analyzed
how one can reliably send information over a discrete memoryless wiretap
channel. Wyner showed\ that, if Eve intercepts a degraded version of the
intended receiver's (Bob's) signal, the transmitter (Alice) can limit the
information leakage by means of channel coding. The associated notion of
\emph{secrecy capacity} was introduced to characterize the maximum
transmission rate from Alice to Bob, below which Eve is unable to obtain any
information.

Wyner's original work provided the theoretical foundation for keyless
security, namely \emph{physical layer security}. Several studies have been
made to generalize Wyner's wiretap channel model. For example, in \cite%
{Hellman78}, Leung-Yan-Cheong and Hellman studied the Gaussian wiretap
channel and showed that a positive secrecy capacity exists only when Eve's
channel is of lower quality than that of Bob. In \cite{Csiszar78}, the
authors considered a non-degraded version of Wyner's wiretap channel. The
notion of wiretap channel has also been extended to fading channels. For
quasi-static fading channels, the outage probability of secrecy capacity is
derived in \cite{Bloch08}. For the ergodic fading channel, \cite{Liang08}
provides a detailed analysis of secrecy capacity. In \cite{Liang08MAC}, the
secrecy capacity region of a multiple-access channel with confidential
messages is derived. The secure transmissions over multiple-output
multiple-input (MIMO) wiretap channels are studied in \cite{Khisti10MO}.
The achievable average secrecy rate has been widely adopted as a
metric of security \cite{Bloch08,Liang08,Liang08MAC,Goel08,Khisti10MO}.

In the context of wiretap code design, Csisz\'{a}r \cite{Csiszar96} proposed
the \emph{strong secrecy} criterion, i.e., $\lim\limits_{n\rightarrow \infty
}I(\mathbf{u}$; $\mathbf{y)}=0$, which implies that the overall information
leakage between the message $\mathbf{u}$ and Eve's channel output $\mathbf{y}
$ should vanish as the codeword length $n$ tends to infinity. Polar codes
achieving strong secrecy over discrete memoryless channels have been
proposed in \cite{Vardy11}. For Gaussian wiretap channels, nested lattice
codes achieving strong secrecy were proposed in \cite{cong12IT}. In
particular, polar codes in \cite{Vardy11} and lattice codes in \cite%
{cong12IT} were shown to achieve semantic security \cite{Goldwasser82}. In \cite%
{oggerb11}, Oggier \emph{et al. }showed that it is possible to construct
lattice codes that maximizes Eve's error probability.

Instead of only relying on the randomness of communication channels,
physical layer jamming techniques were proposed to increase secrecy rate. In
\cite{Goel08}, Goel and Negi showed that it is possible to align additive
white Gaussian noise (AWGN), so called \textquotedblleft artificial
noise\textquotedblright\ (AN), within the null space of a MIMO channel
between Alice and Bob, thus only Eve is jammed. The idea of AN has been
extended to various system models \cite%
{Khisti10SO,Zhou10,Pei12,Liu13Letter,XZhang13}. When the number of Bob's
antennas $N_{\text{B}}$ is one, the asymptotic analysis of the secrecy
capacity was derived in \cite{Khisti10SO}, and its power allocation scheme
was presented in \cite{Zhou10}. In the case of imperfect channel state
information, an average minimum mean square (MSE) uplink-downlink duality
was derived in \cite{Pei12}. More recently, we have shown that Eve's error
probability can be maximized by any randomly distributed AN (not necessarily
Gaussian) \cite{Liu13Letter}. In the existing AN-based schemes, it is
commonly assumed that the number of Eve's antennas $N_{\text{E}}$ is smaller
than the number of Alice's antennas $N_{\text{A}}$, i.e., $N_{\text{E}}<N_{%
\text{A}}$ \cite{Goel08,Zhou10,XZhang13}. However, the relationship between
the secrecy rate and unbounded $N_{\text{E}}$ has never been explored.

In this work, we characterize the exact secrecy rate of the original AN
scheme \cite{Goel08} for any Eve-to-Bob channel noise-power ratios using
various AN power allocation schemes, with arbitrary number of antennas at
each terminal. Our contributions are three-fold: firstly, we derive a
closed-form expression for the \emph{average secrecy rate}; secondly, we
derive a closed-form expression for the asymptotic \emph{instantaneous
secrecy rate} as the number of antennas in each terminal becomes large; and
finally, we derive lower and upper bounds on the average/instantaneous
secrecy rate, leading to simple sufficient and necessary conditions that
guarantee positive average/instantaneous secrecy rate.

The paper is organized as follows: Section II presents the system model,
followed by the analysis of secrecy rate in Section III. Section IV provides
lower and upper bounds on the secrecy rate. Conclusions are drawn in Section
V. Proofs of the theorems are given in Appendix.

\textit{Notation:} Matrices and column vectors are denoted by upper and
lowercase boldface letters, and the Hermitian transpose, inverse,
pseudoinverse of a matrix $\mathbf{B}$ by $\mathbf{B}^{H}$, $\mathbf{B}^{-1}$%
, and $\mathbf{B}^{\dagger }$, respectively. $|\mathbf{B}|$ denotes the
determinant of $\mathbf{B}$. Let the random variables $\left\{ X_{n}\right\}
$ and $X$ be defined on the same probability space. We write $X_{n}\overset{%
a.s.}{\rightarrow }X$ if $X_{n}$ converges to $X$ almost surely or with
probability one. $\mathbf{I}_{n}$ denotes the identity matrix of size $n$.
An $m\times n$ null matrix is denoted by $\mathbf{0}_{m\times n}$. A
circularly symmetric complex Gaussian random variable $x$ with variance $%
\sigma ^{2}$\ is defined as $x\backsim \mathcal{N}_{\mathbb{C}}(0,\sigma
^{2})$. The real, complex, integer and complex integer numbers are denoted
by $\mathbb{R}$, $\mathbb{C}$, $\mathbb{Z}$ and $\mathbb{Z}\left[ i\right] $%
, respectively. $I(x;y)$ represents the mutual information of two random
variables $x$ and $y$. We use the standard asymptotic notation $f\left(
x\right) =O\left( g\left( x\right) \right) $ when $\lim
\sup\limits_{x\rightarrow \infty }|f(x)/g(x)|<\infty $. $\lceil x\rfloor $
rounds to the closest integer. A central complex Wishart matrix $\mathbf{A}%
\in \mathbb{C}^{m\times m}$ with $n$ degrees of freedom and covariance
matrix $\mathbf{\Sigma }$, is defined as $\mathbf{A}\backsim W_{m}(n$,$%
\mathbf{\Sigma })$. We write $\triangleq $ for equality in definition.

\section{System Model}

We consider secure communications over a three-terminal system, including\ a
transmitter (Alice), the intended receiver (Bob), and an unauthorized
receiver (Eve), equipped with $N_{\text{A}}$, $N_{\text{B}}$, and $N_{\text{E%
}}$ antennas, respectively. The signal vectors received by Bob and Eve are%
\begin{align}
\mathbf{z}& =\mathbf{Hx}+\mathbf{n}_{\text{B}}\text{,}  \label{B1} \\
\mathbf{y}& =\mathbf{G\mathbf{\mathbf{x}}}+\mathbf{n}_{\text{E}}\text{,}
\label{E1}
\end{align}%
where $\mathbf{x}\in \mathbb{C}^{N_{\text{A}}\times 1}$ is the transmit
signal vector, $\mathbf{H}\in \mathbb{C}^{N_{\text{B}}\times N_{\text{A}}}$
and $\mathbf{G}\in \mathbb{C}^{N_{\text{E}}\times N_{\text{A}}}$ are the
respective channel matrices between Alice to Bob and Alice to Eve, and $%
\mathbf{n}_{\text{B}}$, $\mathbf{n}_{\text{E}}$ are AWGN vectors with i.i.d.
entries $\sim \mathcal{N}_{\mathbb{C}}(0$, $\sigma _{\text{B}}^{2})$ and $%
\mathcal{N}_{\mathbb{C}}(0$, $\sigma _{\text{E}}^{2})$. We assume that the
entries of $\mathbf{H}$ and $\mathbf{G}$ are i.i.d. complex random variables
$\sim \mathcal{N}_{\mathbb{C}}(0$, $1)$.

Without loss of generality, we normalize Bob's channel noise variance to
one, i.e.,%
\begin{equation}
\sigma _{\text{B}}^{2}=1\text{,}
\end{equation}%
and accordingly normalize the total average transmission power E$(||\mathbf{x}%
||^{2}) $, as in \cite{Telatar99}.

\subsection{Artificial Noise Scheme}

The AN scheme assumes $N_{\text{B}}<N_{\text{A}}$, in order to ensure that $%
\mathbf{H}$ has a non-trivial null space $\mathbf{Z}=\mbox{null}(\mathbf{H})$
(such that $\mathbf{HZ}=\mathbf{0}_{N_{\text{B}}\times (N_{\text{A}}-N_{%
\text{B}})}$) \cite{Goel08}. Let $\mathbf{H}=\mathbf{U}\mathbf{\Lambda }%
\mathbf{V}^{H}$ be the singular value decomposition (SVD) of $\mathbf{H}$,
then we can write the unitary matrix $\mathbf{V}$ as%
\begin{equation}
\mathbf{V=[V}_{1},\mathbf{Z]}\text{,}
\end{equation}%
where the $N_{\text{B}}$ columns of $\mathbf{V}_{1}$ span the orthogonal
complement subspace to the  null space spanned by the columns of $\mathbf{Z}$.

Using the AN scheme, Alice transmits%
\begin{equation}
\mathbf{x}=\mathbf{V}_{1}\mathbf{u}+\mathbf{Zv}=\mathbf{V}\left[
\begin{array}{c}
\mathbf{u} \\
\mathbf{v}%
\end{array}%
\right] \text{,}  \label{T1}
\end{equation}%
where $\mathbf{u}\in \mathbb{C}^{N_{\text{B}}\times 1}$ is the information
vector and $\mathbf{v}\in \mathbb{C}^{(N_{\text{A}}-N_{\text{B}})\times 1}$
is the \textquotedblleft artificial noise\textquotedblright . For the
purpose of evaluating the achievable secrecy rate, both $\mathbf{u}$ and $%
\mathbf{v}$ are assumed to be Gaussian circularly symmetric random vectors
with i.i.d. complex entries $\sim \mathcal{N}_{\mathbb{C}}(0$, $\sigma _{%
\text{u}}^{2})$ and $\mathcal{N}_{\mathbb{C}}(0$, $\sigma _{\text{v}}^{2})$,
respectively.

Equations (\ref{B1}) and (\ref{E1}) can then be rewritten as%
\begin{align}
\mathbf{z}& =\mathbf{H\mathbf{\mathbf{V}}}_{1}\mathbf{u}+\mathbf{HZv+n}_{%
\text{B}}=\mathbf{H\mathbf{\mathbf{V}}}_{1}\mathbf{u}+\mathbf{n}_{\text{B}}%
\text{,}  \label{sec_mod2} \\
\mathbf{y}& =\mathbf{G\mathbf{\mathbf{V}}}_{1}\mathbf{u}+\mathbf{GZ}\mathbf{v%
}+\mathbf{n}_{\text{E}}\text{.}  \label{Eve_mod2}
\end{align}%
From (\ref{sec_mod2}) and (\ref{Eve_mod2}), we note that
$\mathbf{v}$ only degrades Eve's channel, but does not affect Bob.

In our paper, we assume the worst-case scenario for Alice and Bob described
in \cite{Goel08}:

\begin{itemize}
\item Alice has only the knowledge of $\mathbf{H}$.

\item Eve has the knowledge of $\mathbf{H}$, $\mathbf{G}$, $\mathbf{Z}$ and $%
\mathbf{\mathbf{V}}_{1}$.
\end{itemize}

Different from \cite{Goel08}, we assume no upper bound on $N_{\text{E}}$.

Since $\mathbf{V}$ is a unitary matrix, the total transmission power can be
written as%
\begin{equation}
||\mathbf{x}||^{2}=\left[
\begin{array}{c}
\mathbf{u} \\
\mathbf{v}%
\end{array}%
\right] ^{H}\mathbf{V}^{H}\mathbf{V}\left[
\begin{array}{c}
\mathbf{u} \\
\mathbf{v}%
\end{array}%
\right] =||\mathbf{u}||^{2}+||\mathbf{v}||^{2}\text{.}  \label{Total_power}
\end{equation}

We set the average transmit power constraint $P$,%
\begin{equation}
P=\text{E}(||\mathbf{x}||^{2})=P_{\text{u}}+P_{\text{v}}\text{,}
\label{Power_Cons}
\end{equation}%
where%
\begin{equation}
\begin{array}{lll}
P_{\text{u}} &=&\text{E}(||\mathbf{u}||^{2})=\sigma _{\text{u}}^{2}N_{\text{B%
}}\text{,}   \\
P_{\text{v}} &=&\text{E}(||\mathbf{v}||^{2})=\sigma _{\text{v}}^{2}(N_{\text{%
A}}-N_{\text{B}})\text{,}
\end{array}
 \label{Power_c2}
\end{equation}
are fixed by the power allocation scheme that selects the balance between $\sigma _{\text{u}%
}^{2} $ and $\sigma _{\text{v}}^{2}$.

\subsection{Instantaneous and Average Secrecy Capacities}

The idea underpinning the AN scheme is to increase secrecy capacity by
jamming Eve. We recall from \cite{Oggier11} the definition of instantaneous secrecy capacity:%
\begin{equation}
C_{\text{S}}\triangleq \max_{p\left( \mathbf{u}\right) }\left\{ I(\mathbf{%
u;z)-}I(\mathbf{u;y)}\right\} \text{.}  \label{CS}
\end{equation}%
where the maximum is taken over all possible input distributions $p\left(
\mathbf{u}\right) $.

We remark that $C_{\text{S}}$ is a function of $\mathbf{H}$ and $\mathbf{G}$%
, which are embedded in $\mathbf{z}$ and $\mathbf{y}$. To average out the
randomness of $C_{\text{S}}$, we further define the average secrecy
capacity, as in \cite{Goel08}%
\begin{equation}
\bar{C}_{\text{S}}\triangleq \max_{p\left( \mathbf{u}\right) }\left\{ I(%
\mathbf{u;z|H)-}I(\mathbf{u;y|H}\text{, }\mathbf{G)}\right\} \text{,}
\label{CS_ave}
\end{equation}%
where $I\left( X;Y|Z\right) \triangleq \text{E}_{Z}\left[ I\left( X;Y\right)
|Z\right] $, following the notation in \cite{Telatar99}.

Since closed form expressions for $C_{\text{S}}$ and $\bar{C}_{\text{S}}$
are not always available (except for the following Theorem \ref{Th_3} given
in Sec. III.E), we often resort to the corresponding secrecy rates, given by%
\begin{equation}
R_{\text{S}}\triangleq I(\mathbf{u;z)-}I\mathbf{(\mathbf{u;y)}}\text{,}
\label{CLB}
\end{equation}%
\begin{equation}
\bar{R}_{\text{S}}\triangleq I(\mathbf{u;z|H)-}I\mathbf{(\mathbf{u;y|H},%
\mathbf{G)}}\text{,}  \label{CLB_ave}
\end{equation}%
assuming Gaussian input alphabets, i.e., $\mathbf{v}$ and $\mathbf{u}$ are
mutually independent Gaussian vectors with i.i.d. complex entries $\mathcal{N%
}_{\mathbb{C}}(0$, $\sigma _{\text{v}}^{2})$ and $\mathcal{N}_{\mathbb{C}}(0$%
, $\sigma _{\text{u}}^{2})$, respectively.


\subsection{System Parameters}

We define Bob's and Eve's SNRs as

\begin{itemize}
\item SNR$_{\text{B}}\triangleq \sigma _{\text{u}}^{2}/\sigma _{\text{B}%
}^{2} $

\item SNR$_{\text{E}}\triangleq \sigma _{\text{u}}^{2}/\sigma _{\text{E}%
}^{2} $
\end{itemize}

To simplify our notation, we define three system parameters:

\begin{itemize}
\item $\alpha \triangleq \sigma _{\text{u}}^{2}/\sigma _{\text{E}}^{2}$ (SNR$%
_{\text{E}}$)

\item $\beta \triangleq \sigma _{\text{v}}^{2}/\sigma _{\text{u}}^{2}$ (AN
power allocation)

\item $\gamma \triangleq \sigma _{\text{E}}^{2}/\sigma _{\text{B}}^{2}$
(Eve-to-Bob noise-power ratio)
\end{itemize}

Note that SNR$_{\text{B}}=\alpha \gamma $. If $\gamma >1$, we say Eve has a
\emph{degraded} channel. Since we have normalized $\sigma _{\text{B}}^{2}$ to one,
we \ can write (\ref{Power_c2}) as

\begin{itemize}
\item $P_{\text{u}}=\alpha \gamma N_{\text{B}}$

\item $P_{\text{v}}=\alpha \beta \gamma (N_{\text{A}}-N_{\text{B}})$
\end{itemize}

\section{Secrecy Rate with Gaussian Input Alphabets}

In this section, we first derive a closed-form expression for the average
secrecy rate in (\ref{CLB_ave}) with Gaussian input alphabets. We then
present an asymptotic analysis on the instantaneous secrecy rate in (\ref%
{CLB}). Finally, we show average secrecy capacity in (\ref{CS_ave}) is
achieved with Gaussian input alphabets when $N_{\text{E}}\leq N_{\text{A}%
}-N_{\text{B}}$. To present our result, we define some useful functions.

\subsection{Definitions}

We first define the following function (see \cite{shin03})%
\begin{align}
& \Theta (m,n,x)\triangleq
e^{-1/x}\sum_{k=0}^{m-1}\sum_{l=0}^{k}\sum_{i=0}^{2l}\Bigg\{\dfrac{%
\displaystyle(-1)^{i}(2l)!(n-m+i)!}{\displaystyle2^{2k-i}l!i!(n-m+l)!}
\notag \\
& \cdot \left( \!\!\!%
\begin{array}{c}
2(k-l) \\
k-l%
\end{array}%
\!\!\!\right) \cdot \left( \!\!\!%
\begin{array}{c}
2(l+n-m) \\
2l-i%
\end{array}%
\!\!\!\right) \cdot \sum_{j=0}^{n-m+i}x^{-j}\Gamma (-j,1/x)\Bigg\}\text{,}
\label{Cap_CF}
\end{align}%
where
\scalebox{1}[0.8]{$\Big(\!\!\!
\begin{array}{c}
a \\
b\end{array}\!\!\!\Big)$} $=a!/((a-b)!b!)$ is the binomial coefficient, $%
n\geq m$ are positive integers, and $\Gamma (a,b)$ is the incomplete Gamma
function%
\begin{equation}
\Gamma (a,b)=\int_{b}^{\infty }x^{a-1}e^{-x}dx\text{.}
\end{equation}

We further define%
\begin{eqnarray}
N_{\min } &\triangleq &\min \left\{ N_{\text{E}}\text{, }N_{\text{A}}-N_{%
\text{B}}\right\} \text{,}  \label{n_min} \\
N_{\max } &\triangleq &\max \left\{ N_{\text{E}}\text{, }N_{\text{A}}-N_{%
\text{B}}\right\} \text{,}  \label{n_max}
\end{eqnarray}%
\begin{eqnarray}
\hat{N}_{\min } &\triangleq &\min \left\{ N_{\text{E}}\text{, }N_{\text{A}%
}\right\} \text{,}  \label{n_hat_min} \\
\hat{N}_{\max } &\triangleq &\max \left\{ N_{\text{E}}\text{, }N_{\text{A}%
}\right\} \text{.}  \label{n_hat_max}
\end{eqnarray}

Finally, we define a set of $N_{\text{A}}$ power ratios $\left\{ \theta
_{i}\right\} _{1}^{N_{\text{A}}}$, where%
\begin{equation}
\theta _{i}\triangleq \QATOPD\{ . {\alpha ~~~~~~~~~~1\leq i\leq N_{\text{B}%
}}{\alpha \beta ~~~N_{\text{B}}+1\leq i\leq N_{\text{A}}}  \label{theta}
\end{equation}

\subsection{Average Secrecy Rate}

A closed-form expression for $\bar{R}_{\text{S}}$ in (\ref{CLB_ave}) can be
derived using the results from \cite[Th. 2]{Telatar99}, \cite[Th. 1]{shin03}
and \cite[Th. 1]{Chiani10}, leading to the following theorem.

\begin{theo}
\label{Th_R1}%
\begin{equation}
\bar{R}_{\text{S}}=\Theta (N_{\text{B}},N_{\text{A}},\alpha \gamma )+\Theta
(N_{\min },N_{\max },\alpha \beta )-\Omega \text{,}  \label{RS_exact}
\end{equation}%
where $\Theta (\cdot ,\cdot ,\cdot )$ is given in (\ref{Cap_CF}),%
\begin{equation}
\Omega =\left\{
\begin{array}{l}
K\sum\limits_{k=1}^{\hat{N}_{\min }}\det \left( \mathbf{R}^{(k)}\right)
\text{, }~~~~~~\beta \neq 1 \\
\Theta (\hat{N}_{\min },\hat{N}_{\max },\alpha )\text{, }~~~~~~~\beta =1%
\end{array}%
\right.  \label{PPP}
\end{equation}%
\begin{equation}
K=\frac{(-1)^{N_{\text{E}}(N_{\text{A}}-\hat{N}_{\min })}}{\Gamma _{\hat{N}%
_{\min }}(N_{\text{E}})}\frac{\prod\limits_{i=1}^{2}\mu _{i}^{m_{i}N_{\text{E%
}}}}{\prod\limits_{i=1}^{2}\Gamma _{m_{i}}(m_{i})\prod\limits_{i<j}\left(
\mu _{i}-\mu _{j}\right) ^{m_{i}m_{j}}}\text{,}  \label{K}
\end{equation}%
\begin{equation*}
\Gamma _{k}(n)=\prod\limits_{i=1}^{k}(n-i)!\text{,}
\end{equation*}%
and $\mu _{1}>\mu _{2}$ are the two distinct eigenvalues of the matrix $%
\mathrm{diag}\left( \left\{ \theta _{i}^{-1}\right\} _{1}^{N_{\text{A}%
}}\right) $, with corresponding multiplicities $m_{1}$ and $m_{2}$ such that
$m_{1}+m_{2}=N_{\text{A}}$. The matrix $\mathbf{R}^{(k)}$ has elements%
\begingroup
\renewcommand*{\arraystretch}{2.5}
\begin{equation}
r_{i,j}^{(k)}=\left\{\!\!\!\!
\begin{array}{l}
\displaystyle(\mu _{\displaystyle e_{i}})^{N_{\text{A}}-j-d_{i}}\dfrac{%
\displaystyle\left(\displaystyle N_{\text{A}}-j\right) !}{(N_{\text{A}%
}-j-d_{i})!}\text{, }~~~~~\hat{N}_{\min }+1\leq j\leq N_{\text{A}} \\
\displaystyle(-1)^{\displaystyle d_{i}}\dfrac{\displaystyle\varphi (i,j)!}{%
\displaystyle(\mu _{\displaystyle e_{i}})^{\displaystyle\varphi (i,j)+1}}%
\text{,}~~~~~~~~~~~~1\leq j\leq \hat{N}_{\min }\text{, }j\neq k \\
\displaystyle(-1)^{\displaystyle d_{i}}\varphi (i,j)!e^{\displaystyle \mu _{%
\displaystyle e_{i}}}\sum\limits_{l=0}^{\displaystyle\varphi (i,j)}\dfrac{%
\displaystyle\Gamma (l-\varphi (i,j),\mu _{\displaystyle e_{i}})}{%
\displaystyle(\mu _{\displaystyle e_{i}})^{\displaystyle l+1}}\text{, }\text{%
otherwise}%
\end{array}%
\right.  \label{R}
\end{equation}%
\endgroup
where%
\begin{equation*}
e_{i}=\QATOPD\{ . {1~~~~~~~~~~1\leq i\leq m_{1}}{2~~~m_{1}+1\leq i\leq N_{%
\text{A}}}
\end{equation*}%
\begin{equation*}
d_{i}=\sum_{k=1}^{e_{i}}m_{k}-i\text{,}
\end{equation*}%
\begin{equation*}
\varphi (i,j)=N_{\text{E}}-\hat{N}_{\min }+j-1+d_{i}\text{.}
\end{equation*}
\end{theo}

\begin{proof}
See Appendix A.
\end{proof}

Theorem \ref{Th_R1} gives the exact value of $\bar{R}_{\text{S}}$ for the AN
scheme, as a function of SNR$_{\text{B}}$ ($\alpha \gamma $), SNR$_{\text{E}%
} $ ($\alpha $), power allocation scheme ($\beta $), $N_{\text{A}}$, $N_{%
\text{B}}$ and $N_{\text{E}}$. Note that (\ref{RS_exact}) can be expressed
in terms of a finite number of incomplete Gamma functions, thus provides a
closed-form expression for $\bar{R}_{\text{S}}$.

\subsection{Asymptotic Instantaneous\emph{\ }Secrecy Rate}

The instantaneous secrecy rate $R_{\text{S}}$ in (\ref{CLB}) is a function
of $\mathbf{H}$, $\mathbf{V}$ and $\mathbf{G}$. Since we assumed that the
realizations of $\mathbf{H}$, $\mathbf{V}$ and $\mathbf{G}$ are known to
Eve, she is able to compute the exact value of $R_{\text{S}}$. However,
Alice only knows the realizations of $\mathbf{H}$ and $\mathbf{V}$, and
can only assume that $\mathbf{G}$ is a random matrix.
Therefore, $R_{\text{S}}$ is a random variable from Alice's perspective.
We will consider this point of view when designing the secure communications system
and we will use random matrix
theory to characterize the asymptotic behavior (in terms of number of antennas) of the \emph{normalized}
instantaneous secrecy rate $R_{\text{S}}/N_{\text{B}}$.
We then show by simulation that the asymptotic behavior is a very accurate approximation
even for very small numbers of antennas.

The following theorem proves that  $R_{\text{S}}/N_{\text{B}}$ converges to a constant value,
which depends only of the system parameters: $\alpha$, $\beta$, $\gamma$, $P_u$, $P_v$,
and the asymptotic number of antenna ratios.
A special case of this result for $\beta =1$ was given in \cite{Shuiyin13_2}. Here, we
provide a unified result for arbitrary $\beta $.

\begin{theo}
\label{Th_R3}As $N_{\text{A}}$, $N_{\text{B}}$, $N_{\text{A}}-N_{\text{B}}$
and $N_{\text{E}}\rightarrow \infty $ with $N_{\text{A}}/N_{\text{E}%
}\rightarrow \beta _{1}$, $N_{\text{A}}/N_{\text{B}}\rightarrow \beta _{2}$
and $N_{\text{B}}/N_{\text{E}}\rightarrow \beta _{3}$,%
\begin{eqnarray}
\frac{R_{\text{S}}}{N_{\text{B}}} &\overset{a.s.}{\rightarrow }& \Phi (P_{%
\text{u}}\text{, }\beta _{2})-\frac{1}{\beta _{3}}\left( \beta _{1}\mathcal{V%
}(\delta )-\log \delta +\delta -1\right)  \notag \\
&&+\frac{1}{\beta _{3}}\Phi \left( \frac{P_{\text{v}}}{\gamma (\beta
_{1}-\beta _{3})}\text{, }\beta _{1}-\beta _{3}\right) \;\;\triangleq \;\;\Psi \text{%
,}  \label{R3}
\end{eqnarray}%
where%
\begin{eqnarray}
\Phi \left( x\text{, }y\right) &\mathbf{=}&y\log \left( 1+x-\frac{1}{4}%
\mathcal{F}\left( x,y\right) \right) -\frac{\mathcal{F}\left( x,y\right) }{4x%
} \notag  \\
&&+\log \left( 1+xy-\frac{1}{4}\mathcal{F}\left( x,y\right) \right) \text{,}
\label{phi}
\end{eqnarray}%
\begin{equation}
\mathcal{F}\left( x,y\right) \mathbf{=}\left( \sqrt{x\left( 1+\sqrt{y}%
\right) ^{2}+1}-\sqrt{x\left( 1-\sqrt{y}\right) ^{2}+1}\right) ^{2}\text{,}
\label{F}
\end{equation}%
\begin{equation}
\mathcal{V}(\delta )=\frac{1}{\beta _{2}}\log \left( 1+\frac{\delta P_{\text{%
u}}}{\gamma \beta _{3}}\right) +\left( 1-\frac{1}{\beta _{2}}\right) \log
\left( 1+\frac{\delta P_{\text{v}}}{\gamma (\beta _{1}-\beta _{3})}\right)
\text{,} \label{V_delta}
\end{equation}%
and $\delta $ is the solution of the equation
\begin{equation}
\beta _{1}=\frac{1-\delta }{1-\eta (\delta )}\text{,} \label{eq_delta}
\end{equation}%
\begin{equation}
\eta (\delta )\mathbf{=}\frac{1}{\beta _{2}}\left( 1+\frac{\delta P_{\text{u}%
}}{\gamma \beta _{3}}\right) ^{-1}+\left( 1-\frac{1}{\beta _{2}}\right)
\left( 1+\frac{\delta P_{\text{v}}}{\gamma (\beta _{1}-\beta _{3})}\right)
^{-1}\text{.}
\label{eta_delta}
\end{equation}
\end{theo}

\begin{proof}
See Appendix B.
\end{proof}

In the special case of $\beta =1$, according to the definitions of $P_{\text{u}}$ and $P_{\text{v}%
}$ given in Sec. II.C, we have%
\begin{equation}
\dfrac{P_{\text{u}}}{\gamma \beta _{3}}=\dfrac{P_{\text{v}}}{\gamma (\beta
_{1}-\beta _{3})}\text{.}  \label{Eq_beta}
\end{equation}%
According to \cite[Eq. 2.121]{Verdu04}, the explicit solution of (\ref{eq_delta})  is%
\begin{equation}
\delta =1-\dfrac{\mathcal{F}\left( \dfrac{P_{\text{u}}}{\gamma \beta _{3}}%
,\beta _{1}\right) }{\dfrac{4P_{\text{u}}}{\gamma \beta _{3}}}\text{.}
\label{delta}
\end{equation}%
By substituting (\ref{Eq_beta}) and (\ref{delta}) into (\ref{R3}), we have%
\begin{equation}
\dfrac{R_{\text{S}}}{N_{\text{B}}}\overset{a.s.}{\rightarrow }\Phi (P_{\text{%
u}}\text{, }\beta _{2})-\dfrac{\Phi (P_{\text{u}}/(\gamma \beta _{3})\text{,
}\beta _{1})}{\beta _{3}}+\dfrac{\Phi (P_{\text{u}}/(\gamma \beta _{3})\text{%
, }\beta _{1}-\beta _{3})}{\beta _{3}}\text{,}
\end{equation}%
which coincides with \cite[Th. 3]{Shuiyin13_2}.

\subsection{Asymptotic Approximation of Average Secrecy Rate}

Theorem \ref{Th_R3} shows that the random variable $R_{\text{S}}/N_{\text{B}}
$ converges almost surely to a constant $\Psi $ given in (\ref{R3}), as the number
of antennas at each terminal goes to infinity. Hence, also the average normalized secrecy rate
converges to the same constant, as stated in the following corollary.

\begin{corollary}
\label{col1}Under the same assumptions of Theorem \ref{Th_R3},%
\begin{equation}
\frac{\bar{R}_{\text{S}}}{N_{\text{B}}}\rightarrow \Psi \text{,}
\end{equation}%
where $\Psi $ is given in (\ref{R3}). Then we can write that%
\begin{equation}
\bar{R}_{\text{S}}\thickapprox N_{\text{B}}\Psi \triangleq \bar{R}_{\text{%
S,Asym}}\text{.}  \label{R_asym}
\end{equation}
for a sufficiently large number of antennas.
\end{corollary}

\begin{proof}
The proof is straightforward.
\end{proof}

Corollary \ref{col1} provides an alternative way to evaluate $\bar{R}_{\text{%
S}}$. To use the asymptotic approximation $\bar{R}_{\text{S,Asym}}$ for a
finite system model, we substitute in $\Psi $%
\begin{equation}
\beta _{1}=N_{\text{A}}/N_{\text{E}}\text{, }\beta _{2}=N_{\text{A}}/N_{%
\text{B}}\text{, and }\beta _{3}=N_{\text{B}}/N_{\text{E}}\text{.}
\label{Finite}
\end{equation}

\begin{remark}
For finite system models, we verified by simulations that the difference between $\bar{R}_{%
\text{S}}$ and $\bar{R}_{\text{S,Asym}}$ is indistinguishable if%
\begin{equation}
\min \left\{ N_{\text{A}},N_{\text{B}},N_{\text{A}}-N_{\text{B}},N_{\text{E}%
}\right\} >2.
\end{equation}
\end{remark}

\begin{example}
Let us apply Theorem \ref{Th_R1} and Corollary \ref{col1} to the analysis of
an AN scheme with $N_{\text{A}}=6$, $N_{\text{B}}=3$ and $\alpha =3$ dB. Fig.%
$~$1 shows the value of $\bar{R}_{\text{S}}$ with $\beta =-3$ dB, as a
function of $\gamma $ and $N_{\text{E}}$. In Fig.$~$2, we fix $\gamma =3$ dB
and verify the relationship between $\beta $, $N_{\text{E}}$ and $\bar{R}_{%
\text{S}}$. Both figures show an excellent agreement between the
theoretically derived $\bar{R}_{\text{S}}$ and Monte Carlo simulation and
the asymptotic approximation $\bar{R}_{\text{S,Asym}}$.

Moreover, Fig.$~$1 shows how $\bar{R}_{\text{S}}$ increases with
increasing $\gamma $ and decreases with increasing $N_{\text{E}}$. In Fig.$~$%
2, we observe that increasing $\beta $ (i.e., increasing AN power) has
little effect on increasing $\bar{R}_{\text{S}}$ when $N_{\text{E}}>N_{\text{%
A}}$.
\end{example}

\begin{figure}[tbp]
\centering\includegraphics[scale=0.55]{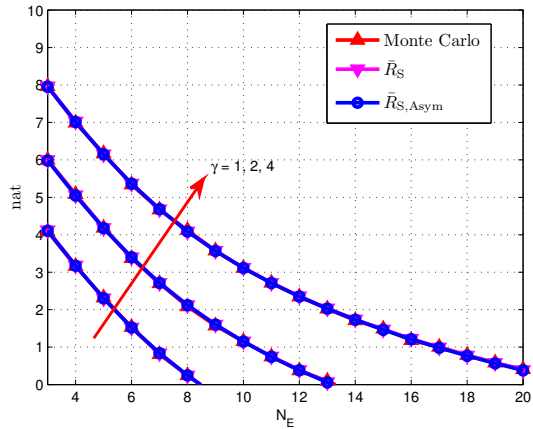}
\caption{$\bar{R}_{\text{S}}$ vs. $\protect\gamma $ and $N_{\text{E}}$ with $%
\protect\alpha =3$ dB, $\protect\beta =-3$ dB, $N_{\text{A}}=6$ and $N_{%
\text{B}}=3$. }
\end{figure}
\begin{figure}[tbp]
\centering\includegraphics[scale=0.55]{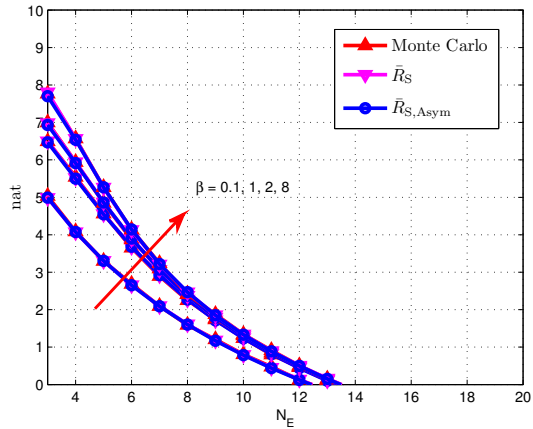}
\caption{$\bar{R}_{\text{S}}$ vs. $\protect\beta $ and $N_{\text{E}}$ with $%
\protect\alpha =3$ dB, $\protect\gamma =3$ dB, $N_{\text{A}}=6$ and $N_{%
\text{B}}=3$. }
\end{figure}

\subsection{Achieving Average Secrecy Capacity}

The following theorem gives a sufficient condition for the achievability of the average secrecy capacity (%
\ref{CS_ave}) using Gaussian input alphabets.

\begin{theo}
\label{Th_3}If $N_{\text{E}}\leq N_{\text{A}}-N_{\text{B}}$, as $\alpha $, $%
\beta \rightarrow \infty $, then%
\begin{equation}
\bar{C}_{\text{S}}=\bar{R}_{\text{S}}=\bar{C}_{\text{Bob}}\text{,}
\label{t3}
\end{equation}%
where $\bar{C}_{\text{Bob}}$ represents Bob's average channel capacity.
\end{theo}

\begin{proof}
See Appendix C.
\end{proof}

According to (\ref{CS_ave}), a universal upper bound on the average MIMO
secrecy capacity is given by%
\begin{equation}
\bar{C}_{\text{S}}\leq \max_{p\left( \mathbf{u}\right) }\left\{ I(\mathbf{%
u;z|H)}\right\} =\bar{C}_{\text{Bob}}\text{.}  \label{t31}
\end{equation}

\begin{remark}
Combining (\ref{t3}) and (\ref{t31}), we can show that the
 \emph{maximum} average MIMO secrecy capacity is achieved by using
the AN transmission scheme and Gaussian input alphabets, if $N_{\text{E}}$ is
not larger that $N_A-N_B$.
\end{remark}

\section{Lower and Upper bounds on Secrecy Rate}

To guarantee positive average/instantaneous secrecy rate, in this section,
we present simple sufficient and necessary conditions, based upon lower and
upper bounds on the average/instantaneous secrecy rate using Gaussian input
alphabets.

\subsection{Bounds on Average Secrecy Rate}

The following theorem bounds $\bar{R}_{\text{S}}$ given in (\ref{RS_exact}).

\begin{theo}
\label{Th_4}%
\begin{equation}
\bar{R}_{\text{LB}}\leq \bar{R}_{\text{S}}\leq \bar{R}_{\text{UB}}\text{,}
\label{AR_B}
\end{equation}%
where the equality holds if $\beta =1$,%
\begin{eqnarray}
\bar{R}_{\text{LB}} &=&\Theta (N_{\text{B}},N_{\text{A}},\alpha \gamma
)+\Theta (N_{\min },N_{\max },\alpha \beta )  \notag \\
&&-\Theta (\hat{N}_{\min },\hat{N}_{\max },\theta _{\max })\text{,}
\label{R_LB}
\end{eqnarray}%
\begin{eqnarray}
\bar{R}_{\text{UB}} &=&\Theta (N_{\text{B}},N_{\text{A}},\alpha \gamma
)+\Theta (N_{\min },N_{\max },\alpha \beta )  \notag \\
&&-\Theta (\hat{N}_{\min },\hat{N}_{\max },\theta _{\min })\text{,}
\label{R_UB}
\end{eqnarray}%
\begin{eqnarray}
\theta _{\min } & \triangleq &\min \{\alpha ,\alpha \beta \}\text{,}  \notag \\
\theta _{\max } & \triangleq &\max \{\alpha ,\alpha \beta \}\text{,}
\end{eqnarray}%
and $\Theta (\cdot ,\cdot ,\cdot )$ is given in (\ref{Cap_CF}).
\end{theo}

\begin{proof}
See Appendix D.
\end{proof}

\begin{example}
Fig.~3 compares the values of $\bar{R}_{\text{S}}$, $\bar{R}_{\text{UB}}$, $%
\bar{R}_{\text{LB}}$ as functions of $N_{\text{E}}$ with $\alpha =3$ dB, $%
\beta =1$ dB, $\gamma =6$ dB, $N_{\text{B}}=3$ and $N_{\text{A}}=4$.
Note that the upper and lower bounds are become tighter as $\beta$ approaches 0 dB.
\end{example}

\begin{figure}[tbp]
\centering\includegraphics[scale=0.55]{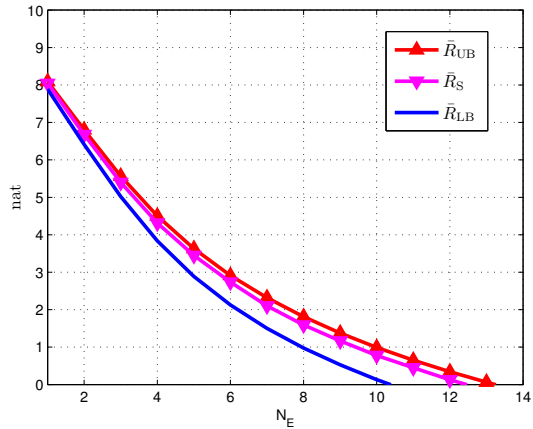}
\caption{$\bar{R}_{\text{S}}$ and $\bar{R}_{\text{LB}}$ vs. $N_{\text{E}}$
with $\protect\alpha =3$ dB, $\protect\beta =1$ dB, $\protect\gamma =6$ dB, $%
N_{\text{B}}=3$ and $N_{\text{A}}=4$.}
\end{figure}

\subsection{Bounds on Instantaneous Secrecy Rate}

We then provide lower and upper bounds on the instantaneous secrecy rate $R_{%
\text{S}}$ in (\ref{CLB}) for high SNR$_{\text{B}}$.

\begin{theo}
\label{Th_5}
Let $N_{\text{A}}$, $N_{\text{B}}$, $N_{\text{A}}-N_{\text{B}}$
and $N_{\text{E}}\rightarrow \infty $ with $N_{\text{A}}/N_{\text{E}%
}\rightarrow \beta _{1}$, $N_{\text{A}}/N_{\text{B}}\rightarrow
\beta _{2}$, $N_{\text{B}}/N_{\text{E}}\rightarrow \beta _{3}$, and
let $P_u,P_v\rightarrow \infty $, then almost surely (i.e., with
probability one)
\begin{eqnarray}
R_{\text{S}}/N_{%
\text{B}} &\geq &\Delta (A_{\max }) \text{,} \notag \\
R_{\text{S}}/N_{%
\text{B}} &\leq & \Delta (A_{\min }) \text{,}  \label{ISR_BOUNDS}
\end{eqnarray}%
where the equality (i.e., almost sure convergence to $\Delta(\cdot)$) holds if $\beta =1$,%
\begin{eqnarray}
\rho & \triangleq &\beta _{1}-\beta _{3}\text{,}  \notag \\
A_{\min } & \triangleq &\min \left\{ \frac{P_{\text{v}}}{\gamma \rho},\dfrac{P_{\text{u}%
}}{\gamma \beta _{3}}\right\} \text{,}  \notag \\
A_{\max } & \triangleq &\max \left\{ \dfrac{P_{\text{v}}}{\gamma \rho},\dfrac{P_{\text{u%
}}}{\gamma \beta _{3}}\right\} \text{,}
\end{eqnarray}%
\begin{equation}
\Delta (x) \triangleq \left( \log P_{\text{u}}\beta _{2}-\left( \beta _{2}-1\right)
\log \left( 1-\beta _{2}^{-1}\right) -1\right) -\digamma (x)+\Xi \text{,}
\label{ddd}
\end{equation}%
\begin{equation}
\digamma (x)  \triangleq \left\{ \!\!\!%
\begin{array}{l}
\beta _{2}\left( \log x-\dfrac{1-\beta _{1}}{\beta _{1}}\log \left( 1-\beta
_{1}\right) -1\right) \text{,}~~~~~~~~~\beta _{1}\leq 1 \\
\beta _{3}^{-1}\left( \log x\beta _{1}-\left( \beta _{1}-1\right) \log
\left( 1-\beta _{1}^{-1}\right) -1\right) \text{,}~\beta _{1}>1%
\end{array}%
\right.  \label{FF}
\end{equation}%
\begin{equation}
\Xi  \triangleq \left\{ \!\!\!%
\begin{array}{l}
\left( \beta _{2}-1\right) \!\!\left( \!\log \dfrac{P_{\text{v}}}{\gamma \rho%
}-\dfrac{1-\rho}{\rho}\log \left( 1-\rho\right) -1\!\!\right) \text{,}%
~~~~~~~~~~~\rho\leq 1 \\
\beta _{3}^{-1}\left( \log \dfrac{P_{\text{v}}}{\gamma }-(\rho-1)\log \left(
1-\rho^{-1}\right) -1\right) \text{.}~~~~~~~~\rho>1%
\end{array}%
\right.
\end{equation}
\end{theo}

\begin{proof}
See Appendix E.
\end{proof}

\subsection{Sufficient and Necessary Conditions for $\bar{R}_{\text{S}},R_{%
\text{S}}>0$}

Theorem \ref{Th_5} shows that the random variable $R_{\text{S}}/N_{\text{B}}$
is almost surely bounded by the constant values $\Delta (A_{\max })$ and $\Delta
(A_{\min })$ given in (\ref{NaR_BOUNDS}).
Then the average normalized secrecy rate is also bounded by the same values, as stated in the
following corollary.

\begin{corollary}
\label{col21}Under the same assumptions of Theorem \ref{Th_5},%
\begin{eqnarray}
 \bar{R}_{\text{S}}/N_{\text{B}} &\geq & \Delta (A_{\max })\text{,}     \notag \\
 \bar{R}_{\text{S}}/N_{\text{B}} &\leq & \Delta (A_{\min })\text{,}
 \label{NaR_BOUNDS}
\end{eqnarray}%
where $\Delta (\cdot )$ is given in (\ref{ddd}). The equality holds if $%
\beta =1$.
\end{corollary}

\begin{proof}
The proof is straightforward.
\end{proof}

The bounds in (\ref{ISR_BOUNDS}) and (\ref{NaR_BOUNDS}) enable the following
simple sufficient and necessary conditions for positive
instantaneous and average secrecy rate.

\begin{corollary}
\label{col2}Let $N_{\text{A}}$, $N_{\text{B}}$, $N_{\text{A}}-N_{\text{B}}$
and $N_{\text{E}}\rightarrow \infty $ with $N_{\text{A}}/N_{\text{E}%
}\rightarrow \beta _{1}$, $N_{\text{A}}/N_{\text{B}}\rightarrow
\beta _{2}$ and $N_{\text{B}}/N_{\text{E}}\rightarrow \beta _{3}$.
Then
a sufficient condition for positive instantaneous and average secrecy rate is given by%
\begin{equation}
\lim_{P_{\text{u}},P_{\text{v}}\rightarrow \infty }\Delta (A_{\max })>0\text{%
,}
\end{equation}%
and a necessary condition for positive instantaneous and average secrecy rate is
given by%
\begin{equation}
\lim_{P_{\text{u}},P_{\text{v}}\rightarrow \infty }\Delta (A_{\min })>0\text{%
,}
\end{equation}%
where $\Delta (\cdot )$ is given in (\ref{ddd}).
\end{corollary}

\begin{proof}
The proof is straightforward.
\end{proof}

To use Corollary \ref{col2} for finite (relatively small)
SNR$_{\text{B}}$ and number of antennas, we substitute in $\Delta
(\cdot )$
\begin{equation*}
\beta _{1}=N_{\text{A}}/N_{\text{E}}\text{, }\beta _{2}=N_{\text{A}}/N_{%
\text{B}}\text{, }\beta _{3}=N_{\text{B}}/N_{\text{E}}\text{,}
\end{equation*}%
\begin{equation}
P_{\text{u}}=\alpha \gamma N_{\text{B}}\text{ and
}P_{\text{v}}=\alpha \beta \gamma
(N_{\text{A}}-N_{\text{B}})\text{.}
\end{equation}%
Thus $\Delta (\cdot )$ reduces to a function of $N_{\text{A}}$, $N_{\text{B}}
$, $N_{\text{E}}$, $\alpha $, $\beta $ and $\gamma $.

\begin{remark}
We verified by simulation that Corollary \ref{col2} is accurate for
finite system models, when
\begin{equation*}
\min \left\{ \alpha \gamma ,\alpha \beta \gamma \right\} \geq 4\text{,}
\end{equation*}%
\begin{equation}
\min \left\{ N_{\text{A}},N_{\text{B}},N_{\text{A}}-N_{\text{B}},N_{\text{E}%
}\right\} >2\text{.}
\end{equation}
\end{remark}

\begin{example}
We consider an AN scheme with $N_{\text{A}}=6$, $N_{\text{B}}=3$%
, $\alpha =\gamma =3$ dB and $\beta =1$ dB. Fig.$~$4 shows the value of $%
\bar{R}_{\text{S}}/N_{\text{B}}$, $\Delta (A_{\max })$ and $\Delta (A_{\min
})$ as functions of $N_{\text{E}}$. By direct computation, $\Delta (A_{\max
})>0$ until $N_{\text{E}}>12$ and $\Delta (A_{\min })<0$ when $N_{\text{E}%
}>16$. It was observed experimentally in Fig. 4 that $\bar{R}_{\text{S}}>0$
when $N_{\text{E}}<12$ and $\bar{R}_{\text{S}}=0$ when $N_{\text{E}}>16$.
\end{example}

Compared to the expressions in Theorems \ref{Th_R1} and \ref{Th_R3},
the sufficient and necessary conditions for positive
average/instantaneous secrecy rate in Corollary \ref{col2} are much
easier to compute, and can be used
for system design. For example, from Alice's perspective, given $N_{%
\text{A}}$, $N_{\text{B}}$, $\alpha $, $\beta $ and $\gamma $, she can
easily predict the number of antennas Eve needs to drive the secrecy rate to
zero.

\begin{figure}[tbp]
\centering\includegraphics[scale=0.55]{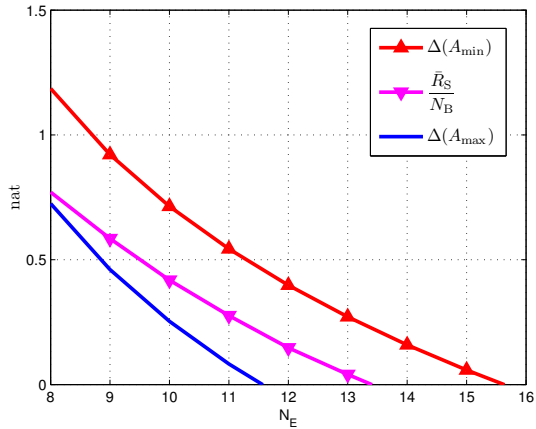}
\caption{Plot of $\bar{R}_{\text{S}}$, $\Delta (A_{\max })$ and $\Delta (A_{\min })$
vs. $N_{\text{E}}$ with $\protect\alpha =\protect\gamma =3$ dB, $\protect%
\beta =1$ dB, $N_{\text{B}}=3$ and $N_{\text{A}}=6$.}
\end{figure}

\section{Conclusions}

In this paper, we have studied the artificial noise scheme in a situation
where Eve has unlimited number of antennas. We derived closed-form
expressions for the average secrecy rate and the asymptotic instantaneous
secrecy rate. The proposed analysis allows a simple evaluation on
the secrecy rate with any SNR$_{\text{B}}$, SNR$_{\text{E}}$, $N_{\text{A}}$%
, $N_{\text{B}}$ and $N_{\text{E}}$, and extends previous studies that were
limited to either the case of $N_{\text{E}}<N_{\text{A}}$ or the case of $N_{%
\text{B}}=1$.

In the second part of this paper, we have derived lower and upper bounds on
the average/instantaneous secrecy rate. Moreover, simple sufficient and
necessary conditions for positive average and instantaneous secrecy rate have
been derived. These easily computable conditions provide Alice and Eve with
design tools for choosing system parameters.

%

\section*{Appendix}

\subsection{Proof of Theorem 1}

Recalling that%
\begin{equation}
\bar{R}_{\text{S}}=I(\mathbf{u;z|H)-}I\mathbf{(\mathbf{u;y|H},\mathbf{G)}}%
\text{.}  \label{CLB_ave_r}
\end{equation}%
In what follows, we evaluate $I(\mathbf{u;z|H)}$ and $I\mathbf{(\mathbf{u;y|H%
},\mathbf{G)}}$, respectively. We treat $\mathbf{H}$ and $\mathbf{G}$ as
Gaussian random matrices.

\emph{1) }$I(\mathbf{u;z\mathbf{|H})}$

Since $\left( \mathbf{HV}_{1}\right) \left( \mathbf{HV}_{1}\right) ^{H}=%
\mathbf{HH}^{H}$, using \cite[Th. 2]{Telatar99} and \cite[Th. 1]{shin03}, we
have%
\begin{eqnarray}
&&I(\mathbf{u;z\mathbf{|H})}  \notag \\
&=&\text{E}_{\mathbf{H}}\left( \log \left\vert \mathbf{I}_{N_{\text{B}%
}}+\alpha \gamma \mathbf{HH}^{H}\right\vert \right)  \notag \\
&=&\Theta (N_{\text{B}},N_{\text{A}},\alpha \gamma )\text{,}  \label{term0}
\end{eqnarray}%
where $\Theta (x,y,z)$ is given in (\ref{Cap_CF}).

\emph{2)} $I\mathbf{(\mathbf{u;y|H},\mathbf{G)}}$

Since all entries in $\mathbf{H}$ and $\mathbf{G}$ are mutually independent,
$I(\mathbf{u;y})$ can be expressed as a function of these independent random
entries. This allows us to take two steps to compute the expected value of $%
I(\mathbf{u;y})$: we first compute $I(\mathbf{u;y\mathbf{|G})}$ given $%
\mathbf{H}$, then compute E$_{\mathbf{H}}\left[ I(\mathbf{u;y|\mathbf{G})}|%
\mathbf{H}\right] $. The advantage is that for given $\mathbf{H}$, $\mathbf{V%
}=[\mathbf{V}_{1}$, $\mathbf{Z}]$ is a fixed unitary matrix. Then, using
\cite[Th. 1]{Lukacs54}, $\mathbf{G\mathbf{V}}_{1}$ and $\mathbf{GZ}$ are
mutually independent complex Gaussian random matrices with i.i.d. entries $%
\sim \mathcal{N}_{\mathbb{C}}(0$, $1)$.

Let $\mathbf{G}_{1}=\mathbf{G\mathbf{V}}_{1}$, $\mathbf{G}_{2}=\mathbf{GZ}$,
$\mathbf{W}_{\text{1}}\mathbf{=\mathbf{G}}_{\text{1}}\mathbf{\mathbf{G}}_{%
\text{1}}^{H}$ and $\mathbf{W}_{\text{2}}\mathbf{=\mathbf{G}}_{\text{2}}%
\mathbf{\mathbf{\mathbf{G}}}_{\text{2}}^{H}$.\ According to \cite{Telatar99}%
, for given $\mathbf{H}$, we have%
\begin{align}
& I(\mathbf{u;y\mathbf{|G})}  \notag \\
& =\text{E}_{\mathbf{G}_{\text{1}}\text{,}\mathbf{G}_{\text{2}}}\left( \log
\frac{\left\vert \mathbf{I}_{N_{\text{E}}}\sigma _{\text{E}}^{2}+\sigma _{%
\text{u}}^{2}\mathbf{W}_{\text{1}}\mathbf{+}\sigma _{\text{v}}^{2}\mathbf{W}%
_{\text{2}}\right\vert }{\left\vert \mathbf{I}_{N_{\text{E}}}\sigma _{\text{E%
}}^{2}\mathbf{+}\sigma _{\text{v}}^{2}\mathbf{W}_{\text{2}}\right\vert }%
\right)  \notag \\
& =\text{E}_{\mathbf{G}_{\text{1}}\text{,}\mathbf{G}_{\text{2}}}\left( \log
\frac{\left\vert \mathbf{I}_{N_{\text{E}}}+\frac{\sigma _{\text{u}}^{2}}{%
\sigma _{\text{E}}^{2}}\mathbf{W}_{\text{1}}\mathbf{+}\frac{\sigma _{\text{v}%
}^{2}}{\sigma _{\text{E}}^{2}}\mathbf{W}_{\text{2}}\right\vert }{\left\vert
\mathbf{I}_{N_{\text{E}}}\mathbf{+}\frac{\sigma _{\text{v}}^{2}}{\sigma _{%
\text{E}}^{2}}\mathbf{W}_{\text{2}}\right\vert }\right)  \notag \\
& =\text{E}_{\mathbf{G}_{\text{1}}\text{,}\mathbf{G}_{\text{2}}}\left( \log
\left\vert \mathbf{I}_{N_{\text{E}}}+\alpha \mathbf{W}_{\text{1}}\mathbf{+}%
\alpha \beta \mathbf{W}_{\text{2}}\right\vert \right) -\text{E}_{\mathbf{G}_{%
\text{2}}}\left( \log \left\vert \mathbf{I}_{N_{\text{E}}}\mathbf{+}\alpha
\beta \mathbf{W}_{\text{2}}\right\vert \right) \text{.}  \label{Ave_I_u_y}
\end{align}

According to \cite[Th. 2]{Telatar99} and \cite[Th. 1]{shin03}, the second
term of (\ref{Ave_I_u_y}) equals to%
\begin{equation}
\text{E}_{\mathbf{G}_{\text{2}}}\left( \log \left\vert \mathbf{I}_{N_{\text{E%
}}}\mathbf{+}\alpha \beta \mathbf{W}_{\text{2}}\right\vert \right) =\Theta
(N_{\min },N_{\max },\alpha \beta )\text{,}  \label{Term2}
\end{equation}%
where $\Theta (x,y,z)$, $N_{\min }$ and $N_{\max }$ are given in (\ref%
{Cap_CF}), (\ref{n_min}) and (\ref{n_max}), respectively.

To compute the first term of (\ref{Ave_I_u_y}), we rewrite $\alpha \mathbf{W}%
_{\text{1}}\mathbf{+}\alpha \beta \mathbf{W}_{\text{2}}$ as $\mathbf{G}_{3}%
\mathbf{\Delta G}_{3}^{H}$, where%
\begin{equation}
\mathbf{G}_{3}=\left[ \mathbf{G}_{1},\mathbf{G}_{2}\right] \text{,}
\label{G3}
\end{equation}%
\begin{equation}
\mathbf{\Delta =}\mathrm{diag}\left( \left\{ \theta _{i}\right\} _{1}^{N_{%
\text{A}}}\right) \text{.}  \label{Cov}
\end{equation}%
$\theta _{i}$ is defined in (\ref{theta}). We define%
\begin{equation}
\Omega =\text{E}_{\mathbf{G}_{\text{1}}\text{,}\mathbf{G}_{\text{2}}}\left(
\log \left\vert \mathbf{I}_{N_{\text{E}}}+\alpha \mathbf{W}_{\text{1}}%
\mathbf{+}\alpha \beta \mathbf{W}_{\text{2}}\right\vert \right) \text{.}
\end{equation}

\emph{Case 1: }If $\beta =1$, using \cite[Th. 2]{Telatar99} and \cite[Th. 1]%
{shin03}, we have%
\begin{equation}
\Omega =\text{E}_{\mathbf{G}_{3}}\left( \log \left\vert \mathbf{I}_{N_{\text{%
E}}}+\alpha \mathbf{G}_{3}\mathbf{G}_{3}^{H}\right\vert \right) =\Theta (%
\hat{N}_{\min },\hat{N}_{\max },\alpha )\text{,}  \label{Term1_beta_1}
\end{equation}%
where $\mathbf{G}_{3}$, $\hat{N}_{\min }$ and $\hat{N}_{\max }$ are given in
(\ref{G3}), (\ref{n_hat_min}) and (\ref{n_hat_max}), respectively.

\emph{Case 2: }If $\beta \neq 1$, $\mathbf{\Delta }^{-1}$ contains two
groups of coinciding eigenvalues. According to \cite[Th. 1]{Chiani10}, we
have%
\begin{equation}
\Omega =\text{E}_{\mathbf{G}_{3}}\left( \log \left\vert \mathbf{I}_{N_{\text{%
E}}}+\mathbf{G}_{3}\mathbf{\Delta G}_{3}^{H}\right\vert \right)
=K\sum\limits_{k=1}^{\hat{N}_{\min }}\det \left( \mathbf{R}^{(k)}\right)
\text{,}  \label{Term1_beta_no1}
\end{equation}%
where $K$ and $\mathbf{R}^{(k)}$ are given in (\ref{K}) and (\ref{R}),
respectively.

Based on (\ref{Ave_I_u_y}), (\ref{Term2}), (\ref{Term1_beta_1}) and (\ref%
{Term1_beta_no1}), we have%
\begin{eqnarray}
&&I(\mathbf{u;y|H,G})  \notag \\
&=&\text{E}_{\mathbf{H}}\left( I(\mathbf{u;y|G})|\mathbf{H}\right)  \notag \\
&=&\Omega -\Theta (N_{\min },N_{\max },\alpha \beta )\text{.}  \label{AAA}
\end{eqnarray}

By substituting (\ref{AAA}) and (\ref{term0}) into (\ref{CLB_ave_r}), we have%
\begin{equation}
\bar{R}_{\text{S}}=\Theta (N_{\text{B}},N_{\text{A}},\alpha \gamma )+\Theta
(N_{\min },N_{\max },\alpha \beta )-\Omega \text{.}
\end{equation}

\QEDA

\subsection{Proof of Theorem 2}

Recalling that%
\begin{equation}
R_{\text{S}}=I(\mathbf{u;z)}-I(\mathbf{u;y)}.
\end{equation}%
In what follows, we evaluate $I(\mathbf{u;z)}$ and $I(\mathbf{u;y)}$,
respectively.

\emph{1) }$I(\mathbf{u;z)}$

Similarly to (\ref{term0}), we have
\begin{equation}
I(\mathbf{u;z)=}\log \left\vert \mathbf{I}_{N_{\text{B}}}+\alpha \gamma
\mathbf{HH}^{H}\right\vert \text{.}  \label{64}
\end{equation}

Since Alice knows the realization of $\mathbf{H}$, she is able to compute
the deterministic value of $I(\mathbf{u;z)}$. As $N_{\text{A}}$ and $N_{\text{B}%
}\rightarrow \infty $ with $N_{\text{A}}/N_{\text{B}}\rightarrow \beta _{2}$%
, the following limit holds \cite[Eq. 1.14]{Verdu04}%
\begin{equation}
\frac{I(\mathbf{u;z)}}{N_{\text{B}}}\rightarrow \Phi \left( P_{\text{u}}%
\text{, }\beta _{2}\right) \text{,}  \label{S1}
\end{equation}%
where $\Phi \left( x\text{, }y\right) $ is given in (\ref{phi}).

\emph{2) }$I(\mathbf{u;y)}$

Similarly to (\ref{Ave_I_u_y}), we have%
\begin{equation}
I(\mathbf{u;y)}=\log \left\vert \mathbf{I}_{N_{\text{E}}}+\alpha \mathbf{W}_{%
\text{1}}\mathbf{+}\alpha \beta \mathbf{W}_{\text{2}}\right\vert -\log
\left\vert \mathbf{I}_{N_{\text{E}}}\mathbf{+}\alpha \beta \mathbf{W}_{\text{%
2}}\right\vert \text{,}  \label{I_simple1}
\end{equation}%
where $\mathbf{G}_{1}=\mathbf{G\mathbf{V}}_{1}$, $\mathbf{G}_{2}=\mathbf{GZ}$%
, $\mathbf{W}_{\text{1}}=\mathbf{\mathbf{G}}_{\text{1}}\mathbf{\mathbf{G}}_{%
\text{1}}^{H}$ and $\mathbf{W}_{\text{2}}=\mathbf{\mathbf{G}}_{\text{2}}%
\mathbf{\mathbf{\mathbf{G}}}_{\text{2}}^{H}$.

From Alice side, $\mathbf{V}=[\mathbf{V}_{1}$, $\mathbf{Z}]$ is a fixed
unitary matrix and $\mathbf{G}$ is a Gaussian random matrix. Using \cite[Th.
1]{Lukacs54}, $\mathbf{G}_{1}$ and $\mathbf{G}_{2}$ are mutually independent
complex Gaussian random matrices with i.i.d. entries $\sim \mathcal{N}_{%
\mathbb{C}}(0$, $1)$.

According to \cite[Eq. 1.14]{Verdu04}, as $N_{\text{A}}-N_{\text{B}}$ and $%
N_{\text{E}}\rightarrow \infty $ with $N_{\text{A}}/N_{\text{E}}\rightarrow
\beta _{1}$ and $N_{\text{B}}/N_{\text{E}}\rightarrow \beta _{3}$, i.e., $%
(N_{\text{A}}-N_{\text{B}})/N_{\text{E}}\rightarrow \beta _{1}-\beta _{3}$,%
\begin{equation}
\frac{\log \left\vert \mathbf{I}_{N_{\text{E}}}\mathbf{+}\alpha \beta
\mathbf{W}_{\text{2}}\right\vert }{N_{\text{E}}}\overset{a.s.}{\rightarrow }%
\Phi \left( \frac{P_{\text{v}}}{\gamma (\beta _{1}-\beta _{3})}\text{, }%
\beta _{1}-\beta _{3}\right) .  \label{S2}
\end{equation}

Let $\mathbf{W}_{3}=\alpha \mathbf{W}_{\text{1}}\mathbf{+}\alpha \beta
\mathbf{W}_{\text{2}}$. We can rewrite $\mathbf{W}_{3}$ as $\mathbf{\hat{G}}%
_{3}\mathbf{T\hat{G}}_{3}^{H}$, where%
\begin{equation}
\mathbf{\hat{G}}_{3}=\left[ \frac{1}{\sqrt{N_{\text{E}}}}\mathbf{G}_{1},%
\frac{1}{\sqrt{N_{\text{E}}}}\mathbf{G}_{2}\right] \text{,}  \label{G3_hat}
\end{equation}%
\begin{equation}
\mathbf{T}=N_{\text{E}}\mathrm{diag}\left( \left\{ \theta
_{i}\right\} _{1}^{N_{\text{A}}}\right) \text{,}  \label{T}
\end{equation}%
where the $\theta_i$'s are given in (\ref{theta}).
Note that the entries of $\mathbf{\hat{G}}_{3}$ are i.i.d. complex random
variables $\sim \mathcal{N}_{\mathbb{C}}(0$, $1/N_{\text{E}})$.

According to \cite[Th. 2.39]{Verdu04}, as $N_{\text{A}}$ and $N_{\text{E}%
}\rightarrow \infty $ with $N_{\text{A}}/N_{\text{E}}\rightarrow \beta _{1}$,%
\begin{equation}
\frac{\log \left\vert \mathbf{I}_{N_{\text{E}}}+\mathbf{W}_{3}\right\vert }{%
N_{\text{E}}}\overset{a.s.}{\rightarrow }\beta _{1}\mathcal{V}(\delta )-\log
\delta +\delta -1\text{,}  \label{S3}
\end{equation}%
where $\mathcal{V}(\delta )$ is given in (\ref{V_delta})
and $\delta $ satisfies%
\begin{equation}
\beta _{1}=\frac{1-\delta }{1-\eta (\delta )}\text{,}
\end{equation}%
with $\eta (\delta )$ given in (\ref{eta_delta}).

From (\ref{S1}), (\ref{S2}) and (\ref{S3}), as $N_{\text{A}}$, $N_{\text{B}}$%
, $N_{\text{A}}-N_{\text{B}}$ and $N_{\text{E}}\rightarrow \infty $ with $N_{%
\text{A}}/N_{\text{E}}\rightarrow \beta _{1}$, $N_{\text{A}}/N_{\text{B}%
}\rightarrow \beta _{2}$ and $N_{\text{B}}/N_{\text{E}}\rightarrow \beta
_{3} $,%
\begin{eqnarray}
&&\frac{R_{\text{S}}}{N_{\text{B}}}\overset{a.s.}{\rightarrow
}\Psi\text{,}
\end{eqnarray}
where the constant $\Psi$ is given in (\ref{R3}).
\QEDA

\subsection{Proof of Theorem 3}

We first show that using Gaussian input alphabets, if $N_{\text{E}}\leq N_{%
\text{A}}-N_{\text{B}}$, $I(\mathbf{u;y\mathbf{|H},\mathbf{G})}\rightarrow 0$
as $\alpha $, $\beta \rightarrow \infty $. We follow the definitions in the
proof of Theorem \ref{Th_R1}. Based on (\ref{Ave_I_u_y}), for a given
realization of $\mathbf{\mathbf{H}}$, we have%
\begin{align}
& I(\mathbf{u;y\mathbf{|G})}  \notag \\
& =\text{E}_{\mathbf{G}_{\text{1}}\text{,}\mathbf{G}_{\text{2}}}\left( \log
\frac{\left\vert \mathbf{I}_{N_{\text{E}}}\sigma _{\text{E}}^{2}+\sigma _{%
\text{u}}^{2}\mathbf{W}_{\text{1}}\mathbf{+}\sigma _{\text{v}}^{2}\mathbf{W}%
_{\text{2}}\right\vert }{\left\vert \mathbf{I}_{N_{\text{E}}}\sigma _{\text{E%
}}^{2}\mathbf{+}\sigma _{\text{v}}^{2}\mathbf{W}_{\text{2}}\right\vert }%
\right)  \notag \\
& \overset{a}{\leq }\text{E}_{\mathbf{G}_{\text{2}}}\left( \log \frac{%
\left\vert \mathbf{I}_{N_{\text{E}}}\sigma _{\text{E}}^{2}+\sigma _{\text{u}%
}^{2}\text{E}_{\mathbf{G}_{\text{1}}}\left( \mathbf{W}_{\text{1}}\right)
\mathbf{+}\sigma _{\text{v}}^{2}\mathbf{W}_{\text{2}}\right\vert }{%
\left\vert \mathbf{I}_{N_{\text{E}}}\sigma _{\text{E}}^{2}\mathbf{+}\sigma _{%
\text{v}}^{2}\mathbf{W}_{\text{2}}\right\vert }\right)  \notag \\
& =\text{E}_{\mathbf{G}_{\text{2}}}\left( \log \frac{\left\vert \mathbf{I}%
_{N_{\text{E}}}+\frac{\sigma _{\text{v}}^{2}}{\sigma _{\text{E}}^{2}+N_{%
\text{B}}\sigma _{\text{u}}^{2}}\mathbf{W}_{\text{2}}\right\vert }{%
\left\vert \mathbf{I}_{N_{\text{E}}}+\frac{\sigma _{\text{v}}^{2}}{\sigma _{%
\text{E}}^{2}}\mathbf{W}_{\text{2}}\right\vert }\right) +N_{\text{E}}\log
\frac{\sigma _{\text{E}}^{2}+N_{\text{B}}\sigma _{\text{u}}^{2}}{\sigma _{%
\text{E}}^{2}}\text{,}  \label{I_boundI}
\end{align}%
where $(a)$ holds due to the concavity of log-determinant function and Jensen's
inequality.

Let%
\begin{equation*}
\mathbf{W} = \QATOPD\{ . {\mathbf{G}_{\text{2}}\mathbf{G}_{\text{2}}^{H}}{%
\mathbf{G}_{\text{2}}^{H}\mathbf{G}_{\text{2}}}%
\begin{array}{c}
\text{if }N_{\text{E}}\leq N_{\text{A}}-N_{\text{B}} \\
\text{if }N_{\text{E}}>N_{\text{A}}-N_{\text{B}}%
\end{array}%
\text{,}
\end{equation*}%
i.e., $\mathbf{W}\sim W_{N_{\min }}(N_{\max }$, $\mathbf{I}_{N_{\min }}%
\mathbf{)}$.

Recalling the definitions of $\alpha $ and $\beta $ in Sec. II.C, and based on
Sylvester's determinant theorem and \cite[Th. 1]{shin03}, the first term of (%
\ref{I_boundI}) can be rewritten as%
\begin{align}
& \text{E}_{\mathbf{G}_{\text{2}}}\left( \log \frac{\left\vert \mathbf{I}%
_{N_{\min }}+\frac{\alpha \beta }{1+\alpha N_{\text{B}}}\mathbf{W}%
\right\vert }{\left\vert \mathbf{I}_{N_{\min }}+\alpha \beta \mathbf{W}%
\right\vert }\right)  \notag \\
& =\Theta (N_{\min },N_{\max },\alpha \beta /(1+\alpha N_{\text{B}}))-\Theta
(N_{\min },N_{\max },\alpha \beta )  \label{GiveH}
\end{align}%
where $\Theta (x,y,z)$ is given in (\ref{Cap_CF}).

From (\ref{I_boundI}) and (\ref{GiveH}), we have%
\begin{eqnarray}
I(\mathbf{u;y|H},\mathbf{G)}&=&
\text{E}_{\mathbf{H}}\left[ I(\mathbf{u;y|\mathbf{G})}|\mathbf{H}\right]
\notag \\
&&\hspace{-20mm}\leq N_{\text{E}}\log (1+\alpha N_{\text{B}})-\Theta (N_{\min },N_{\max
},\alpha \beta )  \notag \\
&&\hspace{-10mm}+\Theta (N_{\min },N_{\max },\alpha \beta /(1+\alpha N_{\text{B}}))  \notag
\\
&&\hspace{-20mm}=(N_{\text{E}}-N_{\min })\log \alpha N_{\text{B}}+O\left( \frac{1}{\alpha }%
\right) +O\left( \frac{1}{\beta }\right) \text{.}  \label{Up_b}
\end{eqnarray}

Based on (\ref{Up_b}), if $N_{\min }=N_{\text{E}}$, i.e., $N_{\text{E}}\leq
N_{\text{A}}-N_{\text{B}}$, as $\alpha $ and $\beta \rightarrow \infty $,%
\begin{equation}
I(\mathbf{u;y|H},\mathbf{G)}=0\text{.}  \label{ub2}
\end{equation}

Under the same conditions, by substituting (\ref{ub2}) into (\ref{CLB_ave}),
we have%
\begin{equation}
\bar{R}_{\text{S}}=I(\mathbf{u;z\mathbf{|H})}=\bar{C}_{\text{Bob}}\text{,}
\label{cc1}
\end{equation}%
where $\bar{C}_{\text{Bob}}$ represents Bob's average channel
capacity. The last equation holds since the input $\mathbf{u}$ is a
circularly symmetric complex Gaussian random vector with zero mean
and covariance $\sigma _{\text{u}}^{2}\mathbf{I}_{N_{\text{B}}}$
\cite[Th. 1]{Telatar99}.

On the other hand, from (\ref{CS_ave}), we have%
\begin{equation}
\bar{C}_{\text{S}}\leq \max_{p\left( \mathbf{u}\right) }\left\{ I(\mathbf{u;z%
\mathbf{|H})}\right\} =\bar{C}_{\text{Bob}}\text{.}  \label{cc2}
\end{equation}

Based on (\ref{cc1}) and (\ref{cc2}), as $\alpha $, $\beta \rightarrow
\infty $, if $N_{\text{E}}\leq N_{\text{A}}-N_{\text{B}}$,%
\begin{equation}
\bar{C}_{\text{S}}=\bar{R}_{\text{S}}=\bar{C}_{\text{Bob}}\text{.}
\end{equation}

\QEDA \vspace{-5 mm}
\subsection{Proof of Theorem 4}

We follow the definitions in the proof of Theorem \ref{Th_R1}. We first
bound the term $\Omega $ in the expression of $\bar{R}_{\text{S}}$ given in (%
\ref{R3}). Recalling that%
\begin{equation}
\Omega =\text{E}_{\mathbf{G}_{\text{1}}\text{,}\mathbf{G}_{\text{2}}}\left(
\log \left\vert \mathbf{I}_{N_{\text{E}}}+\alpha \mathbf{W}_{\text{1}}%
\mathbf{+}\alpha \beta \mathbf{W}_{\text{2}}\right\vert \right) \text{.}
\label{om1}
\end{equation}

Let $\theta _{\min }=\min \{\alpha ,\alpha \beta \}$ and $\theta _{\max
}=\max \{\alpha ,\alpha \beta \}$. Since $\mathbf{W}_{\text{1}}$ and $%
\mathbf{W}_{\text{2}}$ are positive semidefinite matrices, using \cite[Eq.
12, pp. 55]{Helmut96}, we have%
\begin{eqnarray}
&&\left\vert \mathbf{I}_{N_{\text{E}}}+\alpha \mathbf{W}_{\text{1}}\mathbf{+}%
\alpha \beta \mathbf{W}_{\text{2}}\right\vert  \notag \\
&\geq &\left\vert \mathbf{I}_{N_{\text{E}}}+\theta _{\min }(\mathbf{W}_{%
\text{1}}\mathbf{+W}_{\text{2}})\right\vert  \notag \\
&=&\left\vert \mathbf{I}_{N_{\text{E}}}+\theta _{\min }\mathbf{G}_{3}\mathbf{%
G}_{3}^{H}\right\vert \text{,}  \label{OM_L}
\end{eqnarray}%
and%
\begin{eqnarray}
&&\left\vert \mathbf{I}_{N_{\text{E}}}+\alpha \mathbf{W}_{\text{1}}\mathbf{+}%
\alpha \beta \mathbf{W}_{\text{2}}\right\vert  \notag \\
&\leq &\left\vert \mathbf{I}_{N_{\text{E}}}+\theta _{\max }(\mathbf{W}_{%
\text{1}}\mathbf{+W}_{\text{2}})\right\vert  \notag \\
&=&\left\vert \mathbf{I}_{N_{\text{E}}}+\theta _{\max }\mathbf{G}_{3}\mathbf{%
G}_{3}^{H}\right\vert \text{,}  \label{OM_U}
\end{eqnarray}%
where $\mathbf{G}_{3}$, $\hat{N}_{\min }$ and $\hat{N}_{\max }$ are given in
(\ref{G3}), (\ref{n_hat_min}) and (\ref{n_hat_max}), respectively. The
equality holds if $\beta =1$.

Based on (\ref{om1}), (\ref{OM_L}) and (\ref{OM_U}), using \cite[Th. 2]%
{Telatar99} and \cite[Th. 1]{shin03}, we have%
\begin{equation}
\Theta (\hat{N}_{\min },\hat{N}_{\max },\theta _{\min })\leq \Omega \leq
\Theta (\hat{N}_{\min },\hat{N}_{\max },\theta _{\max })\text{,}
\label{Omeg_b}
\end{equation}%
where $\Theta (x,y,z)$ is given in (\ref{Cap_CF}).

By substituting (\ref{Omeg_b}) into (\ref{RS_exact}), we have%
\begin{equation*}
\bar{R}_{\text{LB}}\leq \bar{R}_{\text{S}}\leq \bar{R}_{\text{UB}}\text{,}
\end{equation*}%
where%
\begin{eqnarray}
\bar{R}_{\text{LB}} &=&\Theta (N_{\text{B}},N_{\text{A}},\alpha \gamma
)+\Theta (N_{\min },N_{\max },\alpha \beta )  \notag \\
&&-\Theta (\hat{N}_{\min },\hat{N}_{\max },\theta _{\max })\text{,}
\end{eqnarray}%
\begin{eqnarray}
\bar{R}_{\text{UB}} &=&\Theta (N_{\text{B}},N_{\text{A}},\alpha \gamma
)+\Theta (N_{\min },N_{\max },\alpha \beta )  \notag \\
&&-\Theta (\hat{N}_{\min },\hat{N}_{\max },\theta _{\min })\text{.}
\end{eqnarray}

\QEDA

\subsection{Proof of Theorem 5}

We follow the definitions in the proofs of Theorems \ref{Th_R1} and \ref%
{Th_4}. Based on (\ref{64}), (\ref{I_simple1}), (\ref{OM_L}) and (\ref{OM_U}%
), we have%
\begin{eqnarray}
&&R_{\text{S}}\geq \log \left\vert \mathbf{I}_{N_{\text{B}}}+\alpha \gamma
\mathbf{HH}^{H}\right\vert +\log \left\vert \mathbf{I}_{N_{\text{E}}}\mathbf{%
+}\alpha \beta \mathbf{W}_{\text{2}}\right\vert  \notag \\
&&-\left\vert \mathbf{I}_{N_{\text{E}}}+\theta _{\max }\mathbf{G}_{3}\mathbf{%
G}_{3}^{H}\right\vert \triangleq R_{\text{LB}}\text{,}  \label{in_low}
\end{eqnarray}%
\begin{eqnarray}
&&R_{\text{S}}\leq \log \left\vert \mathbf{I}_{N_{\text{B}}}+\alpha \gamma
\mathbf{HH}^{H}\right\vert +\log \left\vert \mathbf{I}_{N_{\text{E}}}\mathbf{%
+}\alpha \beta \mathbf{W}_{\text{2}}\right\vert  \notag \\
&&-\left\vert \mathbf{I}_{N_{\text{E}}}+\theta _{\min }\mathbf{G}_{3}\mathbf{%
G}_{3}^{H}\right\vert \triangleq R_{\text{UB}}\text{.}  \label{in_up}
\end{eqnarray}%
The equality of the random variables $R_{\text{S}} = R_{\text{LB}}=R_{\text{UB}}$
holds if $\beta =1$, since $\theta _{\min }=\theta _{\max }$.

We then evaluate $R_{\text{LB}}$ and $R_{\text{UB}}$. For convenience, we
define%
\begin{eqnarray*}
\rho &\triangleq&\beta _{1}-\beta _{3}\text{,} \\
A_{\min } &\triangleq&\min \left\{ \frac{P_{\text{v}}}{\gamma \rho},\frac{P_{\text{u}}%
}{\gamma \beta _{3}}\right\} \text{,} \\
A_{\max } &\triangleq&\max \left\{ \frac{P_{\text{v}}}{\gamma \rho},\frac{P_{\text{u}}%
}{\gamma \beta _{3}}\right\} \text{.}
\end{eqnarray*}

Similarly to the proof in Theorem \ref{Th_R3}, using \cite[Eq. 1.14]{Verdu04}%
, as $N_{\text{A}}$, $N_{\text{B}}$, $N_{\text{A}}-N_{\text{B}}$ and $N_{%
\text{E}}\rightarrow \infty $ with $N_{\text{A}}/N_{\text{E}}\rightarrow
\beta _{1}$, $N_{\text{A}}/N_{\text{B}}\rightarrow \beta _{2}$ and $N_{\text{%
B}}/N_{\text{E}}\rightarrow \beta _{3}$,%
\begin{equation}
\frac{\log \left\vert \mathbf{I}_{N_{\text{B}}}+\alpha \gamma \mathbf{HH}%
^{H}\right\vert }{N_{\text{B}}}\overset{a.s.}{\rightarrow }\Phi \left( P_{%
\text{u}}\text{, }\beta _{2}\right) \text{,}  \label{a1}
\end{equation}%
\begin{equation}
\frac{\log \left\vert \mathbf{I}_{N_{\text{E}}}\mathbf{+}\alpha \beta
\mathbf{W}_{\text{2}}\right\vert }{N_{\text{E}}}\overset{a.s.}{\rightarrow }%
\Phi \left( \frac{P_{\text{v}}}{\gamma \rho}\text{, }\rho\right) \text{,}
\label{a2}
\end{equation}%
\begin{equation}
\frac{\left\vert \mathbf{I}_{N_{\text{E}}}+\theta _{\max }\mathbf{G}_{3}%
\mathbf{G}_{3}^{H}\right\vert }{N_{\text{E}}}\overset{a.s.}{\rightarrow }%
\Phi \left( A_{\max }\text{, }\beta _{1}\right) \text{,}  \label{a3}
\end{equation}%
\begin{equation}
\frac{\left\vert \mathbf{I}_{N_{\text{E}}}+\theta _{\min }\mathbf{G}_{3}%
\mathbf{G}_{3}^{H}\right\vert }{N_{\text{E}}}\overset{a.s.}{\rightarrow }%
\Phi \left( A_{\min }\text{, }\beta _{1}\right) \text{,}  \label{a4}
\end{equation}%
where $\Phi \left( x,y\right) $ is given in (\ref{phi}).

We then evaluate the function $\Phi \left( x,y\right) $. From (\ref{phi}),
since $\mathcal{F}\left( x,y\right) =\mathcal{F}\left( xy,y^{-1}\right) $,
it is easy to show the following property%
\begin{equation}
\Phi \left( x,y\right) =\frac{\Phi \left( xy,y^{-1}\right) }{y^{-1}}\text{.}
\label{phi1}
\end{equation}

Recalling \cite[Examples 2.14\&2.15]{Verdu04}, if $y\leq 1$,%
\begin{equation}
\lim_{x\rightarrow \infty }\left( \log x-\frac{\Phi (x\text{, }y)}{y}\right)
=\frac{1-y}{y}\log \left( 1-y\right) +1\text{.}  \label{ex1}
\end{equation}

Using (\ref{phi1}) and (\ref{ex1}), we have the following results on the non random quantities
$\Phi(P_{\text{u}},\beta _{2})$, $\Phi (P_{\text{v}}/(\gamma \rho),\rho)$, $\Phi
\left( A_{\max },\beta _{1}\right) $ and $\Phi \left( A_{\min },\beta
_{1}\right) $.

Since $\beta _{2}>1$, i.e., $N_{\text{A}}>N_{\text{B}}$,%
\begin{eqnarray}
&&\lim_{P_{\text{u}}\rightarrow \infty }\left( \log P_{\text{u}}\beta
_{2}-\Phi (P_{\text{u}}\text{, }\beta _{2})\right)  \notag \\
&=&\lim_{P_{\text{u}}\rightarrow \infty }\left( \log P_{\text{u}}\beta _{2}-%
\frac{\Phi (P_{\text{u}}\beta _{2}\text{, }\beta _{2}^{-1})}{\beta _{2}^{-1}}%
\right)  \notag \\
&=&\left( \beta _{2}-1\right) \log \left( 1-\beta _{2}^{-1}\right) +1\text{.}
\end{eqnarray}

If $\rho\leq 1$,%
\begin{eqnarray}
&&\lim_{P_{\text{v}}\rightarrow \infty }\left( \log \frac{P_{\text{v}}}{%
\gamma \rho}-\frac{\Phi \left( P_{\text{v}}/(\gamma \rho)\text{, }%
\rho\right) }{\rho}\right)  \notag \\
&=&\frac{1-\rho}{\rho}\log \left( 1-\rho\right) +1\text{.}
\end{eqnarray}

If $\rho>1$,%
\begin{eqnarray}
&&\lim_{P_{\text{v}}\rightarrow \infty }\left( \log \frac{P_{\text{v}}}{%
\gamma }-\Phi \left( \frac{P_{\text{v}}}{\gamma \rho}\text{, }\rho\right)
\right)  \notag \\
&=&\lim_{P_{\text{v}}\rightarrow \infty }\left( \log \frac{P_{\text{v}}}{%
\gamma }-\frac{\Phi \left( P_{\text{v}}/\gamma \text{, }\rho^{-1}\right) }{%
\rho^{-1}}\right)  \notag \\
&=&(\rho-1)\log \left( 1-\rho^{-1}\right) +1\text{.}
\end{eqnarray}

If $\beta _{1}\leq 1$,%
\begin{equation}
\lim_{P_{\text{u}},P_{\text{v}}\rightarrow \infty }\left( \log A_{\max }-%
\frac{\Phi \left( A_{\max }\text{, }\beta _{1}\right) }{\beta _{1}}\right) =%
\frac{1-\beta _{1}}{\beta _{1}}\log \left( 1-\beta _{1}\right) +1\text{,}
\end{equation}%
\begin{equation}
\lim_{P_{\text{u}},P_{\text{v}}\rightarrow \infty }\left( \log A_{\min }-%
\frac{\Phi \left( A_{\min }\text{, }\beta _{1}\right) }{\beta _{1}}\right) =%
\frac{1-\beta _{1}}{\beta _{1}}\log \left( 1-\beta _{1}\right) +1\text{.}
\end{equation}

If $\beta _{1}>1$,%
\begin{eqnarray}
&&\lim_{P_{\text{u}},P_{\text{v}}\rightarrow \infty }\left( \log A_{\max
}\beta _{1}-\Phi \left( A_{\max }\text{, }\beta _{1}\right) \right)  \notag
\\
&=&\lim_{P_{\text{u}},P_{\text{v}}\rightarrow \infty }\left( \log A_{\max
}\beta _{1}-\frac{\Phi \left( A_{\max }\beta _{1}\text{, }\beta
_{1}^{-1}\right) }{\beta _{1}^{-1}}\right)  \notag \\
&=&\left( \beta _{1}-1\right) \log \left( 1-\beta _{1}^{-1}\right) +1\text{,}
\end{eqnarray}%
\begin{eqnarray}
&&\lim_{P_{\text{u}},P_{\text{v}}\rightarrow \infty }\left( \log A_{\min
}\beta _{1}-\Phi \left( A_{\min }\text{, }\beta _{1}\right) \right)  \notag
\\
&=&\lim_{P_{\text{u}},P_{\text{v}}\rightarrow \infty }\left( \log A_{\min
}\beta _{1}-\frac{\Phi \left( A_{\min }\beta _{1}\text{, }\beta
_{1}^{-1}\right) }{\beta _{1}^{-1}}\right)  \notag \\
&=&\left( \beta _{1}-1\right) \log \left( 1-\beta _{1}^{-1}\right) +1\text{.}
\end{eqnarray}

Based on the above analysis, as $N_{\text{A}}$, $N_{\text{B}}$, $N_{\text{A}%
}-N_{\text{B}}$, $N_{\text{E}}\rightarrow \infty $ with $N_{\text{A}}/N_{%
\text{E}}\rightarrow \beta _{1}$, $N_{\text{A}}/N_{\text{B}}\rightarrow
\beta _{2}$ and $N_{\text{B}}/N_{\text{E}}\rightarrow \beta _{3}$, and for
$P_{\text{u}},P_{\text{v}}\rightarrow \infty$
\begin{eqnarray}
 R_{\text{LB}}/N_{\text{B}} &\overset{a.s.}{\rightarrow }& \Delta (A_{\max }) \text{,}  \notag \\
 R_{\text{UB}}/N_{\text{B}} &\overset{a.s.}{\rightarrow }& \Delta (A_{\min }) \text{,}  \label{LP}
\end{eqnarray}
where $\Delta (\cdot )$ is given in (\ref{ddd}).

Using (\ref{in_low}), (\ref{in_up}) and (\ref{LP}), under the above
conditions, we have (\ref{ISR_BOUNDS}).
\QEDA


\end{document}